\documentclass[12pt,twoside]{article}

\usepackage[all]{xy}
\usepackage{amsmath,amsthm,amssymb,graphicx,tikz,fancyhdr,hyperref}
\usepackage{graphicx,hyperref,fancyhdr,amsmath,amssymb,mathrsfs,color,xspace,pdflscape}
\hypersetup{colorlinks=true,linkcolor=red,citecolor=blue,filecolor=magneta,urlcolor=cyan}
\usepackage{fancyhdr}
\usepackage{geometry}
\newtheorem{dfn}{Definition}[section]
\newtheorem{thm}[dfn]{Theorem}

\newtheorem{lem}[dfn]{Lemma}

\renewenvironment{proof}{\noindent{\bf Proof. }  \rm } {\hfill{$\Box$}\\}
\numberwithin{equation}{section}

\begin{document}
\thispagestyle{empty}

\fancyhead[LE]{A new integer-valued AR(1) process based on power series thinning operator}

\fancyhead[RO]{E. Mahmoudi, A. Rostami and R. Roozegar}

\begin{center}
\vspace*{4cm}
{\Large \bf A new integer-valued AR(1) process based on power series thinning operator
}\\
{
E. Mahmoudi, A. Rostami, R. Roozegar \\
Department of Statistics, Faculty of Mathematical Sciences, \\
Yazd University, Yazd, Iran\\
}
\end{center}
\begin{abstract}
In this paper, we introduce the first-order integer-valued autoregressive (INAR(1)) model, with Poisson-Lindley innovations based on power series thinning operator. Some mathematical features of this process are given and estimating the parameters is discussed by three methods; conditional least squares, Yule-Walker equations and conditional maximum likelihood.Then the results are studied for three special cases of power series operators. Finally, some numerical results are presented with a discussion to the obtained results and Four real data sets are used to show the potentially of the new process.
\end{abstract}

\section{Introduction}	
In the last few decades, discrete valued time series have been played an important role in scientific research. Many time series in practice have a discrete nature, such as: the number of daily accident on the roads, the number of reserved rooms at a hotel for several days, the number of accidents on a free way every day, the number of chromosome interchanges in cells, the number foggy days, the number of bases of DNA sequences, and so on.

Integer-valued time series to model count data are encountered in many context, therefore, the study and analysis of such count time series is important and motivates a novel research branch with many practical applications. many authors have been analyze integer-valued time series. Jacobs and Lewis \cite{Jacobs1,Jacobs2,Jacobs3} presented the DARMA models. The INAR(1) process were introduced by Mckenzie \cite{McKenzie1985} and AL-Osh and Alzaid \cite{Al-Osh} based on thinning operator. Rist\'{i}c et al. \cite{Rist2009} introduced the geometric first-order integer-valued autoregressive (NGINAR(1)) process with geometric marginal distribution. Recently, Aghababaei Jazi et al. \cite{Aghababaei2012a} discussed a new stationary first-order integer-valued autoregressive process with zero-inflated Poisson innovations (ZINAR(1)). Aghababaei Jazi et al. \cite{Aghababaei2012b}  proposed the geometric INAR(1) process with geometric innovations (INARG(1)). Schweer and Wei\textit{B} \cite{Schweer} introduced a first-order non-negative integer-valued autoregressive process with compound-Poisson innovations (CPINAR(1)) based on the binomial thinning operator. Among models based on the generalizations of the binomial thinning operator, we cite Aly and Bouzar \cite{Aly1994a} and Rist\'{i}c et al. \cite{Rist2012}.

In real-life situations, there are time series of equi-dispersion, over-dispersion and under-dispersion count data. For over-dispersed count data, the integer-valued AR(1) models have been introduced not only based on the over-dispersed marginal distribution but also on the over-dispersed innovations. For example compound Poisson INAR(1) processes: stochastic properties and testing for over-dispersion  \cite{Schweer}, First-order mixed integer-valued  autoregressive processes with zero-inflated generalized power series innovations \cite{Li}, First-order integer-valued AR processes with zero-inflated Poisson innovations \cite{Aghababaei2012a} and integer-valued AR(1) with geometric innovations \cite{Aghababaei2012b}.\\
Recently, Mohammadpour et al.  \cite{Mohammadpour} proposed a first-order integer-valued autoregressive process with Poisson-Lindley marginals based on the binomial thinning. The
innovation structure form this model is complex, consequently and the conditional probabilities of this model do not have a simple form. Also, L\'{i}vio et al. introduced an new INAR(1) model with Poisson-Lindley innovations based on the binomial thinning operator, denoted by INARPL(1) model, for modelling non-negative integer-valued time series with over-dispersion.

In this paper, we propose a new stationary INAR(1) process for modelling count time series based on the power series thinning operator under Poisson-Lindley innovations. We will provide a comprehensive account of the mathematical properties of the proposed  new process. Using the power series distribution as a thinning operator has the advantage that this operator contains the Poisson, binomial, begative binomial and geometric operators as a special case and by fitting this operator to the count time series data, one can obtain the results of these four operators, simultaneously.

The motivation for such process arises from its potential in modelling and analyzing non-negative integer-valued time series when there is an indication of over-dispersion distributions. The use of innovations that come from the Poisson-Lindley distribution, (i) has many advantages than the other discrete distributions and has many applications in biology, (ii) the Poisson-Lindley distribution belongs to compound Poisson family and has other common properties such as unimodality, over-dispersion, and infinite divisibility, (iii) the Poisson-Lindley distribution can be viewed as mixture of geometric and negative binomial distribution which case the smaller amount of skewness and kurtosis of the Poisson-Lindley distribution than the negative binomial distribution (see Ghitany and Al-Mutairi, 2009 for more details about the Poisson-Lindley distribution), (iv) in biological and medical sciences, the occurrence of successive events is dependent. The Poisson and negative binomial distributions can not give a reasonable fit to the biological and medical count data, because of the equi-dispersion of Poisson and under-dispersion of negative binomial. Instead, the Poisson-Lindley distribution is a good candidate for modelling data in ecology, genetics, biological and medical science because of its over-dispersion property; \cite{Shanker}.

The paper is outlined as follows. In Section 2, after introducing the Poisson-Lindley distribution and power series thinning operator, we introduce a new stationary first-order integer-valued autoregressive process with Poisson-Lindley innovations. Several statistical properties of the new process are outlined in this section. In Section 3 , the estimation methods such as conditional least squares, Yule-Walker and the maximum likelihood are obtained. Three special cases of the proposed model are studied in Section 4. Moreover, some numerical results of the estimators are discussed in Section 5. In Section 6, we provide applications to four real data sets and discuss the obtained results. Finally, Section 7 concludes the paper.


\section{Construction of the model}
In this section we introduce a stationary first-order integer-valued autoregressive process with Poisson-Lindley innovations based on power series operator (PSINARPL(1)).
In this paper, we assume that the innovations of process follow a Poisson-Lindley distribution, so we focus on some properties of the Poisson-Lindley distribution.\\
The random variable $X$ is distributed as Poisson-Lindley distribution if its probability mass function can be written in the form
\begin{equation*}
P(X=x)=\frac{\theta^2(x+\theta+2)}{(\theta+1)^{x+3}},~~~x=0,1,2,...~~~~~~\theta>0,
\end{equation*}
which was introduced firstly by Sankaran \cite{Sankaran}. Expectation and variance of this distribution are given by
\begin{equation*}
E(X)= \frac{\theta + 2}{\theta (\theta + 1)}, ~~~~ Var(X)= \frac{\theta^3 + 4 \theta^2 + 6\theta + 2}{\theta^2 (\theta + 1)^2}.
\end{equation*}
Also, the probability generating function and moment generating function are
\begin{eqnarray*}
\varphi_X (t)&=& \frac{\theta^2}{1+\theta} \left[ \frac{1}{(1 + \theta -t)^2} + \frac{1}{(1 + \theta -t)} \right] , \\
 M_X (t)&=& \frac{\theta^2}{1+\theta} \left[ \frac{1}{(1 + \theta -e^t)^2} + \frac{1}{(1 + \theta -e^t)} \right].
\end{eqnarray*}

Figure \ref{p1} shows the pmf of the PL distribution for different values $\theta$.
\begin{figure}[!t]
\centering
\includegraphics[height=10.5cm,width=13cm]{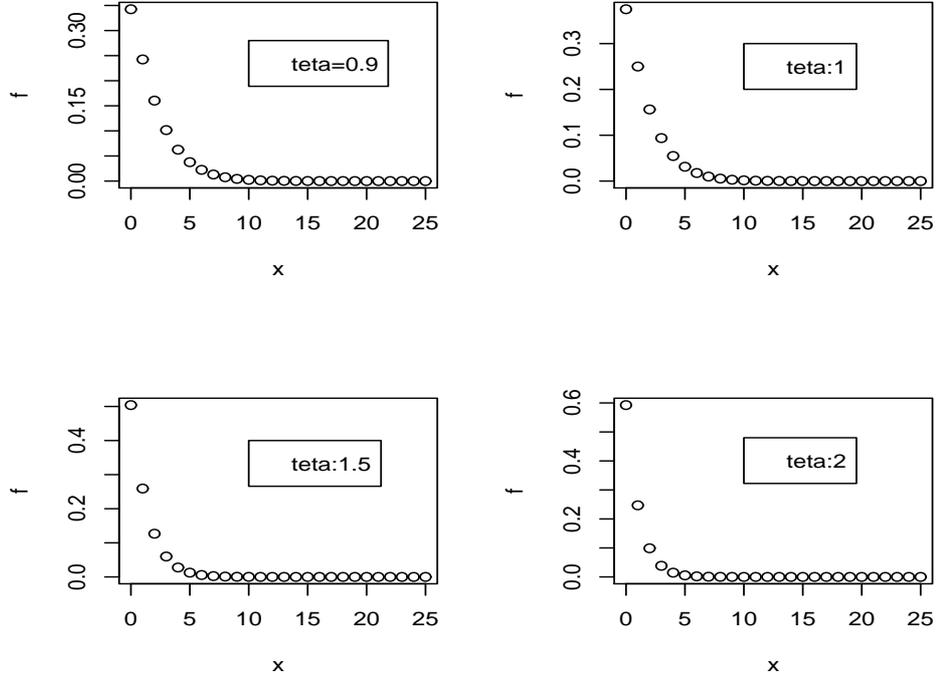} \vspace{-1cm}
\caption{\footnotesize{Probability mass function of PL distribution for different values $\theta$
}} \label{p1}
\end{figure}

The random variable $Y$ with probability mass function
\begin{equation} \label{f2}
 P(Y=y)= \frac{a (y) \beta ^y }{C(\beta)}, \qquad y \in T,
\end{equation}
has power series distribution with range $T$, where $T$ is a subset of the non-negative integer numbers and $C(\beta )=\sum_T a (y) \beta^y $ is finite for all $\beta \in (0, t)$, $a (y) > 0$. The expectation, variance and probability generating function of power series distribution are
$E(Y)= \beta G^{'} (\beta)$, $Var(Y)= \beta G^{'} (\beta) + \beta^2 G^{''} (\beta)$ and $\varphi_Y (s) = \dfrac{C(s \beta)}{C( \beta )} $, where $ G(\beta)=\log C(\beta ) $, $G^{'} (\beta) =\frac{d}{d \beta} G (\beta)  $ and $G^{''} (\beta) =\frac{d^2}{d \beta ^2} G (\beta)$.\\
The following table presents the quantities of the power series distribution family with respect to $\beta$, $C(\beta )$, $a (y) $, $t$ and $T$.

\begin{table}[!h]
{\footnotesize \begin{center}
\begin{tabular}{|l|c|c|c|c|c|} \hline
Distribution & $\beta$ & $C(\beta )$ &  $a (y)$ & $T$  & $t$ \\
\hline
Binomial with parameters $(n,p)$ & $\frac{p}{1-p} $ & $(1+\beta )^n $ & ${{n}\choose{y}}$  & $\{ 0,1,2,...,n \}$  & $\infty $  \\
\hline
Poisson with parameter $\lambda $ & $\lambda$ & $e ^{\beta }$ & $\frac{1}{y! }$ & $\{0,1,... \}$  & $\infty $  \\
\hline
Geometric with parameter  $p$ & $1-p$ & $(1- \beta )^{-1}$ & $1$ & $\{0,1,... \}$  & $1$  \\
\hline
NB with parameter $(r,p)$ & $1-p$ & $(1- \beta )^{-r}$ & $\frac{\Gamma (r+y) }{y! \Gamma (r)} $ & $\{0,1,... \}$  & $1$  \\
\hline
\end{tabular}
\end{center}}
\end{table}

\begin{dfn} \label{defn1.1} (Power series thinning operator) \\
Assume that $X$ is a non-negative integer-valued random variable. Then for each $ \alpha >0$,  power series thinning operator is defined as
\begin{equation} \label{f1}
 \alpha oX=\sum_{i=1}^X Y_i ,
\end{equation}
where $\{Y_i \}$ is a sequence of independent and identically distributed (i.i.d) power series  random variables that are independent of $X$.
\end{dfn}

\begin{dfn} \label{defn1.4} (Construction of the model based on  power series thinning operator) \\ The first-order integer-valued autoregressive model with PL innovations based on power series thinning  operator (PSINARPL(1))) is defined as
\begin{equation} \label{f3}
X_t = \alpha o X_{t-1} + W_t , \qquad \qquad t \geq 1,
\end{equation}
where $W_t$'s are independent and identically distributed random variables from PL distribution that are independent from $Y_i$'s and also from $X_{t- l}$ for $l \geq 1$. Operator $o$ shows the power series thinning that is introduced in Eq. \eqref{f1}.
Note that  $\alpha \in (0,1)$   satisfies dependence and stationary of $X_t$, whereas $\alpha =0 $, and $\alpha\geq1$ implies independence and non-stationary of $X_t$.
\end{dfn}

\subsection{ Statistical properties of the model}
 \begin{lem} \label{lem1}
The mean and variance of $X_t$, (PSINARPL(1) model), are given respectively by
  \begin{description}
\begin{eqnarray} \label{f4}
E(X_t )=  \dfrac{\theta + 2}{\theta (\theta +1)(1-\alpha )} ,
\end{eqnarray}
\begin{eqnarray} \label{f5}
Var(X_t )=\dfrac{\delta( \theta + 2)}{\theta (\theta +1)(1-\alpha)(1-\alpha^2 )} + \dfrac{\theta^3 +4\theta^2 +6\theta +2}{\theta^2 (\theta +1 )^2  (1-\alpha^2)}.
\end{eqnarray}
\end{description}
\end{lem}
 \begin{proof}
According with the properties of the thinning operators which are presented in \cite{Silva}, we have
\begin{eqnarray*}
E(X_t )&=& E(\alpha o X_{t-1} +W_t ) = E(\alpha o X_{t-1})+E(W_t)\\
&=&\alpha E(X_{t-1})+E(W_t)= \dfrac{\mu_{W_t}}{1-\alpha } = \dfrac{\theta + 2}{\theta (\theta +1)(1-\alpha )} , \\\\
Var(X_t )& =&Var (\alpha o X_{t-1} +W_t ) = Var (\alpha o X_{t-1}) +Var(W_t )\\
&=& E(\alpha o X_{t-1})^2- E^2(\alpha o X_{t-1})+Var(W_t)\\
&=&\alpha ^2 E(X_{t-1})^2+\delta E(X_{t-1})-\alpha^2 E^2(X_{t-1})+Var(W_t)\\
&=& \alpha^2 Var(X_{t-1})+\delta(\dfrac{\mu_{W_t}}{1-\alpha})+\sigma^2_{W_t}\\
&=&\dfrac{(\dfrac{\delta}{1-\alpha}) \mu_{W_t} +\sigma^2_{W_t} }{1-\alpha^2 } =\dfrac{\delta( \theta + 2)}{\theta (\theta +1)(1-\alpha)(1-\alpha^2 )} + \dfrac{\theta^3 +4\theta^2 +6\theta +2}{\theta^2 (\theta +1 )^2  (1-\alpha^2)}.\\
\end{eqnarray*}
where $ \mu_{W_t}$ and $\sigma^2_{W_t}$ are the mean and variance of $W_t$'s and $\delta $ is the variance of $Y_i$.
 \end{proof}
 Since in the proposed model variance is greater than mean, the model can also be used for over-dispersed count data modelling.
\subsubsection{Autocovariance and autocorrelation functions}
 Autocovariance and autocorrelation functions for model \eqref{f1} are given respectively by
 \begin{eqnarray*}
 \gamma_k = Cov(X_{t-k},X_t ) &=& Cov(\alpha oX_{t-1}+W_t  ,X_{t-k} )\\
 &=&Cov(\alpha oX_{t-1},X_{t-k} )+Cov(W_t ,X_{t-k} ) \\
 &=& Cov(\alpha oX_{t-1},X_{t-k} )=\alpha \gamma_{k-1}=\alpha ^k \gamma_0 ,
 \end{eqnarray*}
 and $\rho_k = \frac{\gamma_k }{\gamma_0} = \frac{a^k \gamma_0}{\gamma_0} = \alpha ^k $.

 \subsubsection{  Conditional mean and conditional variance }
 The conditional mean and  conditional variance of model \eqref{f1}  are given respectively by
 \begin{eqnarray*}
E(X_{t+1} | X_t )&=&E(\alpha o X_t+W_{t+1} |X_t)\\
&=&E(\alpha o X_t |X_t)+E(W_{t+1}|X_t)=\alpha X_t +\mu_{W_t } ,
 \end{eqnarray*}
and
 \begin{eqnarray*}
Var(X_{t+1}|X_t ) &=& Var(\alpha oX_t+W_{t+1}|X_t ) \\
 &=&  Var(\alpha oX_t|X_t )+Var(W_t|X_t )=\delta X_t+\sigma ^2_{W_t}.
\end{eqnarray*}

\section{Estimation of model parameters}
 Suppose that $X_1,..., X_T , T \in N $ as the time series data, are given. The parameters $\alpha$ and $\theta$ are estimated by the following three methods. To estimate parameter $\theta$, we use the auxiliary parameter $\mu$.

 \subsection{Conditional least squares method (CLS)}
 Conditional least squares estimators of parameters $\alpha$ and $\mu$ for model \eqref{f3} is obtained by minimizing the function
 $$S_n (\alpha,\mu)=\sum_{t=2}^n (X_t -E(X_t |X_{t-1} ))^2 = \sum_{t=2}^n (X_t - \alpha X_{t-1}-(1-\alpha) \mu )^2 ,$$
 where $\mu =E(X_t )$. Thus the conditional least squares estimator of parameters $\alpha$ and $\mu$ for model \eqref{f3} are given as follow,

 $$\hat{\alpha}_{cls} = \frac{(T-1) \sum_{t=2}^T X_t X_{t-1} -\sum_{t=2}^T X_t  \sum_{t=2}^T X_{t-1} }{(T-1)\sum_{t=2}^T X_{t-1}^2 - ( \sum_{t=2}^T X_{t-1} ) }, $$
 $$\hat{\mu}_{cls} = \frac{ \sum_{t=2}^T X_t - \hat{\alpha}_{cls} \sum_{t=2}^T X_{t-1} }{(1-\hat{\alpha}_{cls})(T-1)}. $$
 Also the estimator for $\theta$ is obtained by solving the equation
 $$\hat{\mu}_{cls} = \frac{  \theta + 2 }{\theta(\theta + 1)(1-\hat{\alpha}_{cls})}. $$

 So, $\hat{\theta }_{cls}$ is given by
 $$\hat{\theta }_{cls} = \frac{(1-(1-\hat{\alpha}_{cls})\hat{\mu}_{cls})+ \sqrt{((1-\hat{\alpha}_{cls})\hat{\mu}_{cls} -1)^2 + 8(1-\hat{\alpha}_{cls}) \hat{\mu}_{cls}} }{2(1-\hat{\alpha}_{cls})\hat{\mu}_{cls}}.$$
 \begin{thm} \label{thm1}
The estimators $\hat{\alpha}_{cls}$ and $\hat{\theta}_{cls}$  are strongly consistent for estimating $\alpha$ and $\theta$, respectively, and satisfy the asymptotic normality
\begin{eqnarray}
\sqrt{n}\binom{\hat{\alpha}_{cls}-\alpha}{\hat{\theta}_{cls}-\theta}\to^d N
\begin{pmatrix}
0, c^2 A \\
\end{pmatrix},
\end{eqnarray}
where
\begin{eqnarray*}
A=
\begin{pmatrix}
 r_{11} & r_{12} \\
  r_{21} & r_{22}
\end{pmatrix},
\end{eqnarray*}
and
\begin{eqnarray*}
c &=& \theta^2 (\theta+1)^2[(\mu_2-\mu^2_1)(\theta^2+4\theta+2)]^{-1},
\end{eqnarray*}
\begin{eqnarray*}
 r_{11}&=&\frac{(\theta^2+4\theta+2)^2}{\theta^4(\theta+1)^4}[(\delta\mu_3+\mu_2\sigma ^2_{W_t})-\mu_1(\delta\mu_2+\mu_1\sigma ^2_{W_t})\\
&&+\mu_1(\mu_1(\delta\mu_1+\sigma ^2_{W_t})-(\delta\mu_2+\mu_1\sigma ^2_{W_t}))],
\end{eqnarray*}
\begin{eqnarray*}
 r_{12}&=& r_{21}=\frac{(\theta^2+4\theta+2)}{\theta^2(\theta+1)^2}[\mu_1(\delta\mu_3+\sigma ^2_{W_t}\mu_2)-\mu_2(\delta\mu_2+\mu_1\sigma ^2_{W_t})\\
&&+\mu_1\mu_2(\delta\mu_1+\sigma ^2_{W_t})-\mu^2_1(\delta\mu_2+\mu_1\sigma ^2_{W_t})],
\end{eqnarray*}
\begin{eqnarray*}
r_{22}&=& \mu^2_1(\delta\mu_3+\sigma ^2_{W_t}\mu_2)-2\mu_1\mu_2(\delta\mu_2+\mu_1\sigma ^2_{W_t})+\mu^2_2(\delta\mu_1+\sigma ^2_{W_t}),
\end{eqnarray*}
$E(X_t^r)=\mu_r,~~~ r=1,2,3$
\end{thm}
\begin{proof}
 The proof is similar to proof of Theorem 3.1 of L\'{i}vio et al. \cite{Lívio}, using this fact that in the proof we will consider $f_{t|t-1}=\delta X_{t}+\sigma ^2_{W_t}$. So the proof is omitted.
\end{proof}
 \subsection{Yule-Walker estimation}
 Since $\mu =E(X_t )$ and $\alpha =\frac{\gamma (1)}{\gamma (0)}$, then in model \eqref{f3}, Yule-Walker estimators of $\alpha$ and  $\mu$ are obtained as follows
 $$\hat{\mu}_{YW} = \bar{X}_T= \frac{1}{T} \sum_{t=1}^T X_t , $$
 $$\hat{\alpha }_{YW} = \frac{\hat{\gamma} (1)}{\hat{\gamma} (0)} = \frac{\sum_{t=2}^T (X_t - \bar{X}_t)((X_{t-1} - \bar{X}_t) }{\sum_{t=1}^T (X_t - \bar{X}_t)^2 }.$$
 The Yule-Walker estimator of $\theta$ is given by
  $$\hat{\theta }_{YW} = \frac{(1-(1-\hat{\alpha}_{YW})\hat{\mu}_{YW})+ \sqrt{((1-\hat{\alpha}_{YW})\hat{\mu}_{YW} -1)^2 + 8(1-\hat{\alpha}_{YW}) \hat{\mu}_{YW}} }{2(1-\hat{\alpha}_{YW})\hat{\mu}_{YW}}.$$

   \subsection{Maximum likelihood estimation method}
 Maximum likelihood estimators of $\alpha$ and $\theta$ are obtained by maximizing the likelihood function
 $$L(\theta ,\alpha|\mathbf{x})=f(x_1 ,..., x_n )=f(x_1 )f(x_2|x_1 )...f(x_n|x_{n-1} ).$$
 Because in general case, obtaining the marginal distribution of $X_1$ is hard, a simple method to find the likelihood function is that we condition on variable $X_1$, such that
\begin{eqnarray*}
 f(x_1 ,...,x_n |x_1 ) &=& \frac{f(x_1 ,...,x_n )}{f(x_1 )}.
\end{eqnarray*}
Thus the maximum likelihood estimators of $\alpha$ and $\theta$ in the  model are obtained by maximizing the conditional likelihood function that, in general, no closed form for the conditional maximum likelihood estimates.

\section{Special cases of PSINARPL(1)}
In this section, some special cases of  PSINARPL(1) process are studied and some properties of the model are obtained.
\subsection{Construction the model based on binomial thinning operator}
\begin{dfn} \label{defn1.4} (BINARPL(1)) \\
The first-order integer-valued autoregressive model with Poisson-Lindley innovations based on binomial thinning operator is defined as follows
\begin{equation} \label{f6}
X_t = \alpha o X_{t-1} + W_t , \qquad \qquad t \geq 1,
\end{equation}
where $\alpha \in (0,1)$, $\alpha o X_{t-1}=\sum_{i=1}^{X_{t-1}} Y_i $  where ${Y_i}$ are iid Bernoulli distribution with probability $\alpha$ and $W_t\sim PL(\alpha,\theta)$ are iid and independent from $Y_i$'s and also from $X_{t-l}$ for $t \geq l$.
\end{dfn}
The binomial thinning operator has been introduced and used by many researches such as; Steutel and van Harn \cite{Steutel}, Alzaid and Al-Osh \cite{Alzaid}, Aghababaei Jazi et al. \cite{Aghababaei2012b} and Mohammadpour et al. \cite{Mohammadpour}. For the last work in this field see L\'{i}vio et al.\cite{Lívio}.\\
Since model BINARPL(1) is a special case of model PSINARPL(1), one can obtain different properties of this model using the general results presented in Sections 2 \& 3. By assuming $Y_i\sim Ber(\alpha)$ and letting $\delta=\alpha(1-\alpha)$, the mean and variance of ${X_t}$ are given by
\begin{eqnarray*}
E(X_t )&=& \dfrac{\theta + 2}{\theta (\theta +1)(1-\alpha )},\\
Var(X_t )& =& \dfrac{\alpha ( \theta + 2)}{\theta (\theta +1)(1-\alpha^2 )} + \dfrac{\theta^3 +4\theta^2 +6\theta +2}{\theta^2 (\theta +1 )^2  (1-\alpha^2)}.
\end{eqnarray*}
Also the conditional expectation and the conditional variance are given by
\begin{eqnarray*}
E(X_{t+1} | X_t )&=&\alpha X_t +\frac{\theta + 2}{\theta (\theta + 1)},\\
 Var(X_{t+1}|X_t )&=&\alpha(1-\alpha) X_t+\frac{\theta^3 + 4 \theta^2 + 6\theta + 2}{\theta^2 (\theta + 1)^2}.
\end{eqnarray*}
Given that PSINARPL(1) is Markov process, thus the transition probabilities are given by
 \begin{eqnarray*}
 P_{lk}&=& P(X_t =k |X_{t-1}=l)=P(\alpha oX_{t-1}+W_t =k|X_{t-1}=l)  \\
 &=& \sum_{m} P(\alpha oX_{t-1}=m|X_{t-1}=l) P(W_t =k-m).
\end{eqnarray*}
Because $T=\{0,1,\cdots,l\}$ for fixed $ l \in Z^{+}$ then the inequalities $0 \le m \le l$ and $k-m\ge 0$  implies $ 0\le m \le \min(l,k)$, so we have
  \begin{eqnarray*}
  P_{lk}&=& \sum_{m=0}^{\min(l,k)} P(\alpha oX_{t-1}=m|X_{t-1}=l) P(W_t =k-m) \\
 &=& \sum_{m=0}^{\min(l,k)}  {{l}\choose{m}}  \alpha^m (1-\alpha)^{l-m} \left[\frac{\theta^2 ((k-m)+\theta +2)}{(\theta+1)^{k-m+3} } I_{ \{ 0,1,... \} } (k-m)\right].
 \end{eqnarray*}
Using the Markov property, the joint probability distribution function is obtained as
  \begin{eqnarray*}
f(j_1 ,...,j_n ) &=& P(X_1=j_1 )P(X_2 =j_2 |X_1=j_1 ) ... P(X_n=j_n |X_{n-1} = j_{n-1)} )  \\
 &=& P_{j1} \prod_{t=1}^{n-1} \sum_{m=0}^{\min(j_t,j_{t+1})} {{j_t}\choose{m}}  \alpha^m (1-\alpha)^{j_t -m} \\
&& \times \left [\frac{\theta^2 (j_{t+1}-m+\theta +2)}{(\theta +1)^{j_{t+1}-m+3}}  I_{ \{0,1,... \} } (j_{t+1}-m)\right].
 \end{eqnarray*}
 Also the marginal distribution is calculated as
  \begin{eqnarray*}
 P_k &=& P(X_t =k) = \sum_{l=0}^{\infty}P_{lk}  P(X_{t-1}=l) \\
 &=&   \sum_{l=0}^{\infty} \sum_{m=0}^{\min(l,k)} {{l}\choose{m}}  \alpha^l (1-\alpha )^{l-m} \left[ \frac{\theta^2 (k-m+\theta +2)}{(\theta +1)^{k-m+3} }   I_{ \{0,1,... \} } (k-m) \right]  P_l.
 \end{eqnarray*}
According with the previous section, the CLS and YW estimators of parameters $\alpha$ and $\theta$ for the BINARPL(1) model can be obtained.
Also the MLE of parameters $\alpha$ and $\theta$ of BINARPL(1) are obtained by maximizing the following conditional likelihood function,
\begin{eqnarray*}
 f(x_1 ,...,x_n |x_1 ) &=& \frac{f(x_1 ,...,x_n )}{f(x_1 )} \\
 &=&  \prod_{i=1}^{n-1} \sum_{m=0}^{\min(x_i ,x_{i+1})} {{x_i}\choose{m}}  \alpha^m (1-\alpha)^{x_i-m} \\
&& \times \left[ \frac{\theta^2 (x_{i+1} -m+\theta +2)}{ (\theta +1)^{x_{i+1}-m+3} } I_{ \{ 0,1,... \} }  (x_{i+1}-m)    \right].
\end{eqnarray*}

\subsection{Construction of the model based on  negative binomial thinning operator}
\begin{dfn} \label{defn1.2} (NBINARPL(1)) \\
The first-order integer-valued autoregressive model with Poisson-Lindley innovations based on negative binomial thinning operator is defined as follows
\begin{equation} \label{f6}
X_t = \alpha o X_{t-1} + W_t , \qquad \qquad t \geq 1,
\end{equation}
where $\alpha \in (0,1)$, $\alpha o X_{t-1}=\sum_{i=1}^{X_{t-1}} Y_i $  where ${Y_i}$ are iid geometric distribution with probability mass function $P(Y_i=y)=\frac{\alpha^y}{(1+\alpha)^{y+1}}$ and $W_t\sim PL(\alpha,\theta)$ are iid and independent from $Y_i$'s and also from $X_{t-l}$ for $t \geq l$.
\end{dfn}
 Rist\'{i}c et al. \cite{Rist2009} introduced a new geometric first-order integer-valued autoregressive (NGINAR(1)) process and a combined geometric INAR(p) model based on negative binomial thinning operator is proposed by Nast\'{i}c et al. \cite{Nasti}. Also one can see Janjic et al. \cite{Janjic} for more properties about the binomial and negative binomial thinning operators.
One can obtain
different properties of NBINARPL(1) model using the general results presented in Sections 2 \&
3. By assuming $Y_i\sim G(1/(1+\alpha))$ and letting $\delta=\alpha(1+\alpha)$, the mean and variance of ${X_t}$ are given by
\begin{eqnarray*}
E(X_t )&=& \dfrac{\theta + 2}{\theta (\theta +1)(1-\alpha )}, \\
Var(X_t ) &=& \dfrac{\alpha ( \theta + 2)}{\theta (\theta +1)(1-\alpha )^2} + \dfrac{\theta^3 +4\theta^2 +6\theta +2}{\theta^2 (\theta +1 )^2  (1-\alpha^2)}.
\end{eqnarray*}

Also the conditional expectation and the conditional variance are given by
 \begin{eqnarray*}
E(X_{t+1} | X_t )&=&\alpha X_t +\frac{\theta + 2}{\theta (\theta + 1)},\\
 Var(X_{t+1}|X_t )&=&\alpha(1+\alpha) X_t+\frac{\theta^3 + 4 \theta^2 + 6\theta + 2}{\theta^2 (\theta + 1)^2}.
\end{eqnarray*}
Transition probabilities of the NBINARPL(1) model are given by
 \begin{eqnarray*}
P_{lk} &=& \sum_{m=0}^k {{l+m-1}\choose{m}} \left( \frac{1}{1+\alpha} \right) ^l \left( \frac{\alpha}{1+ \alpha }   \right) ^m  \left[  \frac{\theta^2 (k-m+\theta +2)}{(\theta +1)^{k-m+3}}  I_{ \{ 0,1,... \} } (k-m)  \right] \\
 & &\times I(l \neq 0) + \left[  \frac{ \theta^2 (k+\theta +2)}{(\theta +1)^{k+3}}  I_{ \{ 0,1,... \} } (k)  \right] I(l = 0).
 \end{eqnarray*}
 The joint probability distribution function is given by
  \begin{eqnarray*}
 f(j_1 ,...,j_n ) &=& P_{j1} \prod_{t=1}^{n-1} \Big( \sum_{m=0}^{j_{t+1}} {{j_t +m-1}\choose{m}}   \left( \frac{1}{1+\alpha }  \right) ^{j_t} \left( \frac{\alpha }{1+\alpha }  \right)^m\\
 &&\times \left[ \frac{ \theta^2 ((j_{t+1}-m)+\theta +2)}{(\theta +1)^{j_{t+1}-m+3}} I_{ \{ 0,1,... \} } (j_{t+1} -m) \right] I(j_t \neq 0 ) \\
 &&+ \left[ \frac{\theta^2 ((j_{t+1} )+\theta +2)}{(\theta+1)^{j_{t+1}+3}} I_{ \{ 0,1,... \} } (j_{t+1} ) \right] I(j_t =0)\Big).
 \end{eqnarray*}
 The marginal distribution of NBINARPL(1) can be calculated as
   \begin{eqnarray*}
P_k &=& \sum_{l=0}^{\infty} \Big(\sum_{m=0}^k  {{l +m-1}\choose{m}}
 \left( \frac{1}{1+\alpha }  \right)^{l} \left( \frac{\alpha}{1+\alpha}  \right) ^m \\
 && \times \left[ \frac{ \theta^2 ((k-m)+\theta +2)}{(\theta +1)^{j_{k-m+3}}} I_{ \{ 0,1,... \} } (k-m) \right]  I(l \neq 0 ) \\
 &&   +\left[ \frac{\theta^2 (k+\theta +2)}{(\theta+1)^{k +3}} I_{ \{ 0,1,... \} } (k ) \right] I(l =0)\Big) P_l .
 \end{eqnarray*}

The MLE of parameters $\alpha$ and $\theta$ of NBINARPL(1) model are obtained by maximizing the following conditional likelihood functions;
\begin{eqnarray*}
f(x_1 ,...,x_n |x_1 ) &=& \prod_{i=1}^{n-1} \Big( \sum_{m=0}^{x_{i+1}} {{x_i +m -1}\choose{m}}  \left( \frac{1}{1 + \alpha }  \right) ^{x_i}   \left( \frac{\alpha }{1 + \alpha }  \right) ^{m}  \\
&& \times \left[ \frac{\theta^2 (x_{i+1} -m+\theta +2)}{ (\theta +1)^{x_{i+1}-m+3}} I_{ \{ 0,1,... \} }  (x_{i+1}-m) \right]  I(x_i \neq 0) \\
&&+ \left[ \frac{\theta^2 (x_{i+1} +\theta +2)}{ (\theta +1)^{x_{i+1}+3} } I_{ \{ 0,1,... \} }  (x_{i+1}) \right]  I(x_i = 0)\Big).
\end{eqnarray*}
\subsection{Construction of the model based on Poisson thinning operator}
\begin{dfn} \label{defn1.1} (PINARPL(1)) \\
The first-order integer-valued autoregressive model with Poisson-Lindley innovations based on Poisson thinning operator is defined as follows
\begin{equation} \label{f6}
X_t = \alpha o X_{t-1} + W_t , \qquad \qquad t \geq 1,
\end{equation}
where $\alpha \in (0,1)$, $\alpha o X_{t-1}=\sum_{i=1}^{X_{t-1}} Y_i $  where ${Y_i}$ are iid Poisson distribution with probability mass function $P(Y_i=y)=\frac{e^{-\alpha}\alpha^y}{y!}$ and $W_t\sim PL(\alpha,\theta)$ are iid and independent from $Y_i$'s and also from $X_{t-l}$ for $t \geq l$.
\end{dfn}
Different properties of PINARPL(1) model can be obtained using this fact that $Y_i\sim Pois(\alpha)$ and $\delta=\alpha$. The mean and variance of ${X_t}$ are given by
\begin{eqnarray*}
E(X_t )&=& \dfrac{\theta + 2}{\theta (\theta +1)(1-\alpha )}, \\
Var(X_t )&=&\dfrac{\alpha( \theta + 2)}{\theta (\theta +1)(1-\alpha)(1-\alpha^2 )} + \dfrac{\theta^3 +4\theta^2 +6\theta +2}{\theta^2 (\theta +1 )^2  (1-\alpha^2)}.
\end{eqnarray*}
The conditional expectation and the conditional variance of PINARPL(1) are given by
\begin{eqnarray*}
E(X_{t+1} | X_t )&=&\alpha X_t +\frac{\theta + 2}{\theta (\theta + 1)},\\
Var(X_{t+1}|X_t )&=&\alpha X_t+\frac{\theta^3 + 4 \theta^2 + 6\theta + 2}{\theta^2 (\theta + 1)^2}.
\end{eqnarray*}
Transition probabilities of the PINARPL(1) model are given by
 \begin{eqnarray*}
P_{lk} &=& \sum_{m=0}^k \frac{e^{-\alpha l}(\alpha l)^m}{m!}   \left[  \frac{\theta^2 (k-m+\theta +2)}{(\theta +1)^{k-m+3}}  I_{ \{ 0,1,... \} } (k-m)  \right] \\
 & &\times I(l \neq 0) + \left[  \frac{ \theta^2 (k+\theta +2)}{(\theta +1)^{k+3}}  I_{ \{ 0,1,... \} } (k)  \right] I(l = 0).
 \end{eqnarray*}
 The joint probability distribution function is given by
  \begin{eqnarray*}
 f(j_1 ,...,j_n )&=& P_{j1} \prod_{t=1}^{n-1} \Big( \sum_{m=0}^{j_{t+1}} \frac{e^{-\alpha j_t}(\alpha j_t)^m}{m!}\\
 &&\times \left[ \frac{ \theta^2 ((j_{t+1}-m)+\theta +2)}{(\theta +1)^{j_{t+1}-m+3}} I_{ \{ 0,1,... \} } (j_{t+1} -m) \right] I(j_t \neq 0 ) \\
 &&+ \left[ \frac{\theta^2 ((j_{t+1} )+\theta +2)}{(\theta+1)^{j_{t+1}+3}} I_{ \{ 0,1,... \} } (j_{t+1} ) \right] I(j_t =0)\Big).
 \end{eqnarray*}
 Also, the marginal distribution of this model can be calculated as
   \begin{eqnarray*}
P_k &=& \sum_{l=0}^{\infty} \Big(\sum_{m=0}^k  \frac{e^{-\alpha l}(\alpha l)^m}{m!}
 \times \left[ \frac{ \theta^2 ((k-m)+\theta +2)}{(\theta +1)^{j_{k-m+3}}} I_{ \{ 0,1,... \} } (k-m) \right]  I(l \neq 0 ) \\
 &&   +\left[ \frac{\theta^2 (k+\theta +2)}{(\theta+1)^{k +3}} I_{ \{ 0,1,... \} } (k ) \right] I(l =0)\Big) P_l .
 \end{eqnarray*}

To obtain the ML estimators of parameters $\alpha$ and $\theta$, we need to maximize the conditional likelihood function. For PINARPL(1) this function is given by
\begin{eqnarray*}
f(x_1 ,...,x_n |x_1 ) &= & \frac{f(x_1 ,...,x_n )}{f(x_1 )}\\
&=& \prod_{i=1}^{n-1} \Big( \sum_{m=0}^{x_{i+1}} \frac{e^{-\alpha x_i}(\alpha x_i)^m}{m!}\\
 &&\times \left[ \frac{\theta^2 (x_{i+1} -m+\theta +2)}{ (\theta +1)^{x_{i+1}-m+3}} I_{ \{ 0,1,... \} }  (x_{i+1}-m) \right]  I(x_i \neq 0) \\
&&+ \left[ \frac{\theta^2 (x_{i+1} +\theta +2)}{ (\theta +1)^{x_{i+1}+3} } I_{ \{ 0,1,... \} }  (x_{i+1}) \right]  I(x_i = 0)\Big).
\end{eqnarray*}

\section{Some numerical results}
In this section, for each three models, BINARPL(1), NBINARPL(1) and PINARPL(1), we produce 1000 samples of size T = 100, 200 and 300 and obtain the estimators of the parameters using three methods that are presented in the previous section, then we compare these estimators together.\\
The average estimators (AE), average bias (ABias) and average root mean square errors (RMSE) are reported in Tables \ref{j1}, \ref{j2} and \ref{j3}. In each three  proposed models estimators converge to the true value and also the RMSE decreases when sample size increases.\\
In two sub-models, BINARPL(1) and NBINARPL(1), RMSE of the maximum likelihood estimators are less than the RMSE of  the CLS and YW estimators. In PINARPL(1) model, RMSE of the CLS and YW estimators of parameter $\alpha$  are less than the RMSE of the  ML estimator while the RMSE of the  ML estimator of parameter $\theta $ is less than the RMSE of the CLS and YW estimators.
\begin{table}[h]
\begin{center}
{\scriptsize
\caption{\footnotesize{ CLS, YW and ML estimators of $\alpha$ and $\theta$  for BINARPL(1) } }\label{j1}
\begin{tabular}{ccccccc} \hline
     Sample size    & $ \hat{\alpha }_{BCLS}$ & $  \hat{\theta }_{BCLS} $ & $  \hat{\alpha }_{SYW} $ & $  \hat{\theta }_{SYW} $ & $  \hat{ \alpha}_{BML}$ & $  \hat{\theta }_{BML} $\\
\hline
True value   $\alpha  = 0.2$ and   $\theta = 0.6$  & & & & & &  \\
T=100  & & & & & &  \\
  AEs &  0.1704        &    0.5826   &         0.1687     &         0.5877     &         0.2040      & 0.6099  \\
 ABias & -0.0295  &  -0.0174   &  -0.0313   &    -0.0122  & 0.0040    &  0.0099 \\
 RMSE & 0.1017       &    0.0802       &      0.1013         &    0.0799     &      0.0582        &      0.0692\\\hline
 T=200  & & & & & & \\
AEs       &    0.1853   &          0.5888 &             0.1845    &       0.5915    &       0.1994   &           0.6026\\
ABias    &    -0.0146      &       -0.0111     &       -0.0155     &     -0.0085      &   -0.0006   &          0.0026\\
RMSE    &     0.0739       &       0.0591    &         0.0738       &     0.0590     &      0.0415    &        0.0474\\\hline
T=300 & & & & & & \\
AEs     &   0.1831       &         0.5867   &         0.1825       &    0.5884     &     0.2011     &        0.6047 \\
ABias    &-0.0169      &         -0.0133     &     -0.0175       &   -0.0116     &    0.0011     &       0.0047 \\
RMSE   & 0.0596    &           0.0473      &      0.0596         &   0.0470    &    0.0361         &   0.0403\\\hline\hline
True value $\alpha=0.5$ and $\theta =1 $& & & & & & \\
T=100 & & & & & & \\
AEs      &      0.4483           &  0.9227&             0.4435 &         0.9308 &          0.4980  &           1.0077\\
ABias      &   -0.0517        &    -0.0772 &           -0.0565  &         -0.0692&         -0.0020  &           0.0077\\
RMSE        & 0.1089       &      0.1732    &          0.1106     &        0.1713  &       0.0542     &         0.1349\\\hline
T=200 & & & & & & \\
AEs         &   0.4591          &     0.9270    &         0.4570     &       0.9312    &      0.4999      &        1.0138\\
ABias  &      -0.0409        &      -0.0730      &       -0.0430      &     -0.0687     &     -0.0001      &      0.0138\\
RMSE    &     0.0779      &        0.1315         &     0.0789          &  0.1300         &     0.0395    &       0.0949\\\hline
T=300  & & & & & & \\
AEs    &    0.4629       &         0.9259       &         0.4612      &      0.9284    &  0.5010     &     1.0077\\
ABias    & -0.0371    &           -0.0741        &      -0.0388        &   -0.0716    &        0.0010 &          0.0077\\
RMSE    & 0.0674   &            0.1178            &     0.0683           & 0.1165       &      0.0313   &       0.0799\\\hline\hline
True value $\alpha=0.9$  and $\theta =2$ & & & & & & \\
T=100 & & & & & & \\
AEs  &     0.8603 &              1.7063                & 0.8504&               1.7235&            0.8974&          2.0404\\
ABias  &   -0.0397  &           -0.2937             &  -0.0495  &            -0.2764  &        -0.0026   &       0.0404\\
RMSE    &  0.0664    &         0.6176            &    0.0728      &          0.6245     &       0.0174     &     0.3322\\\hline
T=200  & & & & & & \\
AEs    &   0.8793   &           1.8186          &      0.8743  &               1.8230   &         0.8992 &       2.0273\\
ABias   & -0.0207     &       -0.1814        &      -0.0257     &           -0.1770      &    -0.0008     &         0.0273\\
RMSE    &0.0421        &     0.5074       &       0.0450          &       0.5118 &      0.0113              &   0.2319\\\hline
T=300  & & & & & & \\
AEs  & 0.8853  & 1.8613  &     0.8821      &   1.8657  &    0.8995   &     2.0232 \\
 ABias  &   -0.0147   &    -0.1387  &-0.0178    &  -0.1342   &  -0.0005   &     0.0232 \\
RMSE   & 0.0314    &0.4179     &  0.0330   &     0.4143  &    0.0095    &      0.1902 \\
\hline
\end{tabular}}
\end{center} \vspace{-2.5cm}
\end{table}

\begin{table}[h]
\begin{center}
{\scriptsize
\caption{\footnotesize{ CLS, YW and ML estimators of  $\alpha$ and $\theta$  for NBINARPL(1)  } }\label{j2}
\begin{tabular}{ccccccc} \hline
     Sample size    & $ \hat{\alpha }_{NBcls}$ & $  \hat{\theta }_{NBcls} $ & $  \hat{\alpha }_{NBYw} $ & $  \hat{\theta }_{NBYw} $ & $  \hat{ \alpha}_{NBML}$ & $  \hat{\theta }_{NBML} $
\\\hline
True value   $\alpha  = 0.2 $ and   $\theta = 0.6 $ & & & & & & \\
T=100   & & & & & & \\
 AEs  &   0.1718     &          0.5839  &           0.1702    &           0.5892 &         0.1993&         0.6084\\
 ABias  &   -0.0282     &         -0.0160 &          -0.0298&              -0.0108 &       -0.0007 &       0.0084\\
RMSE    &      0.1072   &   0.0831     &       0.1065&               0.0829       & 0.0761       &   0.0778\\\hline
T=200   & & & & & & \\
AEs   &    0.1840              &   0.5875 &             0.1832  &            0.5902      &     0.2019&          0.6056\\
ABias   &      -0.0160      &         -0.0124&              -0.0168&          -0.0098 &          0.0019&            0.0056\\
RMSE     &    0.07567   &            0.0610 &             0.0757    &        0.0609&            0.0535  &         0.0546  \\\hline
T=300    & & & & & & \\
AEs    &      0.1832    &                0.5847  &       0.1826&           0.5864     &          0.2017  &         0.6051\\
ABias    &    -0.0168&                 -0.0153    &     -0.0174 &          -0.0135     &       0.0017     &         0.0051\\
RMSE      & 0.0620&                   0.0482       &   0.0620     &       0.0478      &         0.0446      &      0.0427 \\\hline\hline
True value $\alpha$=0.5 and $\theta =1$  & & & & & & \\
T=100   & & & & & & \\
AEs&            0.4359&                0.8859&             0.4314       &      0.8936&              0.4793    &       1.0103\\
ABias&         -0.0641  &           -0.1141   &         -0.0686      &    -0.1064     &       -0.0207      &      0.0103\\
RMSE  &     0.1216       &         0.1823      &         0.1234   &         0.1785      &      0.0944    &          0.1765\\\hline
T=200   & & & & & & \\
AEs       &    0.4504 &          0.9022     &          0.4480             & 0.9059   &         0.4953     &       1.0077\\
ABias       & -0.0496   &        -0.0978     &         -0.0520         &    -0.0941   &       -0.0047    &       0.0077\\
RMSE  &      0.0899      &          0.1485    &           0.0910    &          0.1464   &        0.0641 &           0.1203\\\hline
T=300   & & & & & & \\
AEs       &     0.4584         &     0.9072       &        0.4567  &               0.9096   &       0.4922 &           1.0026\\
ABias       &  -0.0416      &      -0.0927         &     -0.0433&                -0.0904     &    -0.0078 &          0.0026\\
RMSE     &   0.0749       &      0.1312            &  0.0757 &  0.1297            & 0.0541 &          0.1021 \\\hline\hline
True value $\alpha=0.9$  and $\theta =2$  & & & & & & \\
T=100  & & & & & & \\
AEs &           0.8047      &        1.3007 &            0.7935   &            1.2959    &         0.8520  &             1.9061\\
ABias &        -0.0953   &         -0.6993  &           -0.1065    &         -0.7041  &         -0.0480     &          -0.0939\\
RMSE   &     0.1343   &          0.8027      &       0.1418          &      0.7972          &  0.0912         &      0.5500\\\hline
T=200  & & & & & & \\
AEs   &      0.8481       &        1.4356        &         0.8422   &         1.4312       &      0.8803          &    1.9568\\
ABias   &   -0.0519    &          -0.5644         &      -0.0578     &      -0.5687     &      -0.0197  &            -0.0432\\
RMSE     & 0.0788   &              0.6743          &     0.0824        &    0.6719         &   0.0488     &         0.3871\\\hline
T=300   & & & & & & \\
AEs  &          0.8604    &         1.4857 &        0.8571       &     1.4904             &0.8832    &        1.9543\\
ABias  &       -0.0396 &           -0.5143  &            -0.0430  &         -0.5096    &       -0.0168&     -0.0457\\
RMSE    &     0.0620&              0.6030    &           0.0646     &        0.6039 &           0.0391  &    0.3173\\
\hline
\end{tabular}}
\end{center} \vspace{-2.5cm}
\end{table}

\begin{table}[h]
\begin{center}
{\scriptsize
\caption{\footnotesize{ CLS, YW and ML estimators of $\alpha$ and $\theta$  for PINARPL(1) } }\label{j3}
\begin{tabular}{ccccccc} \hline
     Sample size    & $ \hat{\alpha }_{PCLS}$ & $  \hat{\theta }_{PCLS} $ & $  \hat{\alpha }_{PYW} $ & $  \hat{\theta }_{PYW} $ & $  \hat{ \alpha}_{PML}$ & $  \hat{\theta }_{PML} $\\
\hline
True value   $\alpha  = 0.2$ and   $\theta = 0.6$  & & & & & &  \\
T=100  & & & & & &  \\
  AEs &    0.1725      &  0.5844     &0.1708         & 0.5896           &     0.1312          &  0.6324 \\
 ABias &  -0.0275 &  -0.0156   &   -0.0292   &  -0.0104    &  -0.0688   & 0.0324 \\
 RMSE &0.1044      &    0.0820      &    0.1041        & 0.0820        &   0.0962           &  0.0774    \\\hline
 T=200  & & & & & & \\
AEs       & 0.1821    &    0.5880  &   0.1812    &     0.5906  &  0.1271  &   0.6239      \\
ABias    &  -0.0179   & -0.0119   &  -0.0188    &     -0.0094  &       -0.0729        &  0.0239       \\
RMSE    &  0.0756   &  0.0588      &  0.0754   &     0.0586     & 0.0863   &   0.0551     \\\hline
T=300 & & & & & & \\
AEs     &   0.1867     &      0.5892     &     0.1860   &   0.5909     &   0.1267  &   0.6227     \\
ABias    & -0.0132     &   -0.0108       &   -0.0139     &      -0.0091     &    -0.0733    &    0.0227    \\
RMSE   & 0.0607     &  0.0491&0.0606&0.0489&       0.0819         &            0.0456   \\\hline\hline
True value $\alpha=0.5$ and $\theta =1 $& & & & & & \\
T=100 & & & & & & \\
AEs      &   0.4426    &  0.9084     &     0.4384     &      0.9167     &   0.4424  &  1.0785    \\
ABias      & -0.0573    &    -0.0915    &     -0.0616   &          -0.0832   &   -0.0576                  &    0.0785   \\
RMSE        &    0.1175  & 0.1793     &       0.1188    &       0.1768   &   0.1012  &         0.1887        \\\hline
T=200 & & & & & & \\
AEs         & 0.4588      &      0.9172    &    0.4563  &    0.9212     &    0.4498      &       1.0761          \\
ABias  &   -0.0412     &        -0.0828        &       -0.0436       &        -0.0788     &    -0.0502          &   0.0761       \\
RMSE    &   0.0844      &    0.1395         &      0.0852           &      0.1376   &     0.0783      &      0.1387      \\\hline
T=300  & & & & & & \\
AEs    & 0.4595    &       0.9177    &      0.4579       &       0.9203    &   0.4504       &   1.0726     \\
ABias    &       -0.0405    &            -0.0823        &        -0.0421     &         -0.0797       &       -0.0496      &    0.0726        \\
RMSE    &    0.0729   &      0.1255         &      0.0738       &         0.1241       &     0.0700      &      0.1193            \\\hline\hline
True value $\alpha=0.9$  and $\theta =2$ & & & & & & \\
T=100 & & & & & & \\
AEs  &  0.8289    &  1.4464                &  0.8178        &       1.4307        &  0.6528          &   2.4360       \\
ABias  &   -0.0711  &       -0.5535               &  -0.0822 &      -0.5692        &       -0.2472    &   0.4360   \\
RMSE    &  0.1050   &       0.8810             &     0.1127     &       0.6997       &     0.9387       &  1.3185   \\\hline
T=200  & & & & & & \\
AEs    & 0.8572     &      1.5699              &  0.8520      &    1.5733             &       0.7295   &  2.4420     \\
ABias   &  -0.0428   &     -0.4301         &   -0.0479        &     -0.4267            &     -0.1705     &   0.4420      \\
RMSE    &   0.0659   &  0.5841        &      0.0695       &   0.5848     &         0.2041           & 0.5434   \\\hline
T=300  & & & & & & \\
AEs  & 0.8683  &             1.6127               &                 0.8649     &        1.6167    &   0.7427    &    2.4417 \\
 ABias  &   -0.0317                  &        -0.3873            &        -0.0350        &      -0.3830        &      -0.1573       &  0.4417   \\
RMSE   &      0.0505          &        0.5273             &         0.0530              &           0.5322         &      0.1818        &      0.5047     \\
\hline
\end{tabular}}
\end{center}  \vspace{-2.5cm}
\end{table}

 \section{Real data examples}
In this section, to compare the proposed three sub-models together and compare them with integer-valued AR(1) with Poisson innovations  based on binomial operator (INARP(1)) and integer-valued AR(1) with geometric innovations  based on binomial operator (INARG(1)), we apply four real time series data sets.

\subsection{The number of earthquakes per year magnitude 7.0 or greater}
The first example assumes the number of earthquakes per year magnitude 7.0 or greater (1900-1998). Time series plot, autocorrelation and partial autocorrelation functions are shown in Figure \ref{p11}. Sample mean, variance and autocorrelation are respectively, 20.02, 52.75 and 0.58.
 \begin{figure}[!h] \vspace{-0.5cm}
\centering
\includegraphics[height=11cm,width=15cm]{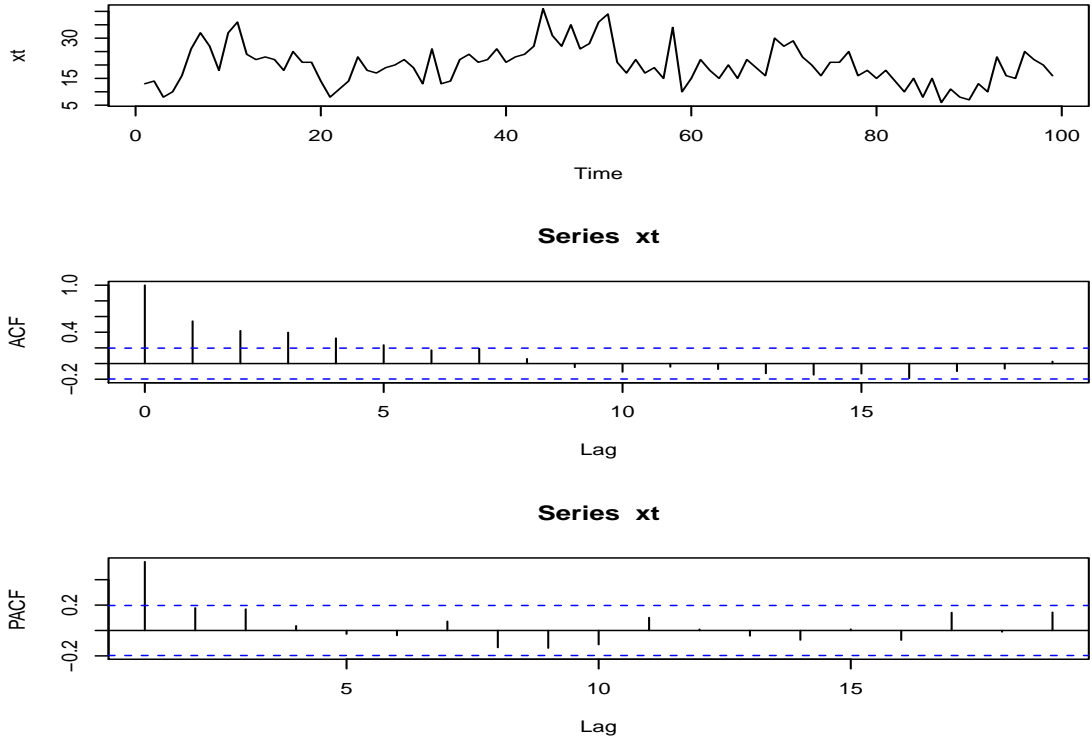} \vspace{-1cm}
\caption{\footnotesize{The time series, ACF and PACF plots of the number of earthquakes per year magnitude 7.0 or greater (1900-1998).}} \label{p11}\vspace{-0.3cm}
\end{figure}

Now, for data modelling, we compare four models. For each model, we calculate the MLE, CLS and YW of parameters, the Akaike information criterion (AIC) and Bayesian information criterion (BIC).
The results are shown in Table \ref{tab1}. According to Table \ref{tab1}, we see that the AIC and BIC of PINARPL(1) is smaller than the AIC and BIC of other models and hence the PINARPL(1) gives the best fit to this data in comparing with NBINARPL(1), BINARPL(1), INARP(1) and INARG(1). So, PINARPL(1) model with $W_t \sim PL(0.2878)$ innovations that gives by
$$X_t =0.694 o X_{t-1}+W_t,$$
is more appropriate for this data. The predicted values of the number of earthquakes per year magnitude 7.0 or greater series are given by
\begin{eqnarray*}
\hat{X}_1 &=& \frac{\hat{\theta} +2 }{\hat{\theta} (\hat{\theta} +1 )(1 - \hat{\alpha})} =20.187,\\
\hat{X}_i &=& \hat{\alpha} \hat{X}_{i-1} + \frac{\hat{\theta}  +2 }{\hat{\theta} (\hat{\theta} +1 )} =
 0.694 \hat{X}_{i-1} + 6.173 ,\quad\quad i=2,3,...,99.
\end{eqnarray*}

\begin{table}[!h]\vspace{-0.5cm}
\begin{center}
{\scriptsize
\caption{\footnotesize{Estimated parameters, AIC and BIC for the number of earthquakes per year magnitude 7.0 or greater.} }\label{tab1}
\begin{tabular}{ccccccc} \hline
Model &      CLS &     YW & MLE &  AIC &   BIC\\
\hline
  NBINARPL(1) &    $  \hat{\alpha } =0.5434 $ &   $  \hat{\alpha } =0.5417$ & $  \hat{\alpha }=0.7398 $   &   637.9338 & 643.1241        \\

 & $  \hat{\theta }  = 0.1969     $ &    $  \hat{\theta }=  0.1998$  & $  \hat{\theta }= 0.3330$ & & \\
&&&&& \\
BINARPL(1) & $  \hat{\alpha } =0.5434 $ &   $  \hat{\alpha } =0.5417$ & $  \hat{\alpha }=0.6099 $   &   642.9801& 648.1704         \\

 & $  \hat{\theta }  = 0.1969     $ &    $  \hat{\theta }=  0.1998$  & $  \hat{\theta }= 0.2304$ & & \\
&&&&& \\
PINARPL(1) & $  \hat{\alpha } =0.5434 $ &   $  \hat{\alpha } =0.5417$ & $  \hat{\alpha }=0.6942 $   &   636.1583& 641.3485        \\

 & $  \hat{\theta }  = 0.1969     $ &    $  \hat{\theta }=  0.1998$  & $  \hat{\theta }= 0.2878$ & & \\
 &&&&& \\
INARG(1) &  $  \hat{\alpha } =0.5434 $ &   $  \hat{\alpha } =0.5417$ & $  \hat{\alpha }=0.7398 $   &   637.9338 & 643.1241        \\

 & $  \hat{P }  = 0.0.097     $ &    $  \hat{P }=  0.1998$  & $  \hat{P }= 0.3330$ & & \\
&&&&& \\
INARP(1) &  $  \hat{\alpha } =0.5434 $ &   $  \hat{\alpha } =0.5417$ & $  \hat{\alpha }=0.3822 $   &   674.5856 & 679.7758        \\

 & $  \hat{\lambda }  = 9.323     $ &    $  \hat{\lambda }=  9.1746$  & $  \hat{\lambda }= 12.42$ & & \\
\hline
\end{tabular}}
\end{center}\vspace{-0.5cm}
\end{table}
Figure \ref{p12} shows the predicted values to the sample paths of earthquakes per year magnitude 7.0 or greater.

 \begin{figure}[!tph]
\centering
\includegraphics[height=9cm,width=9cm]{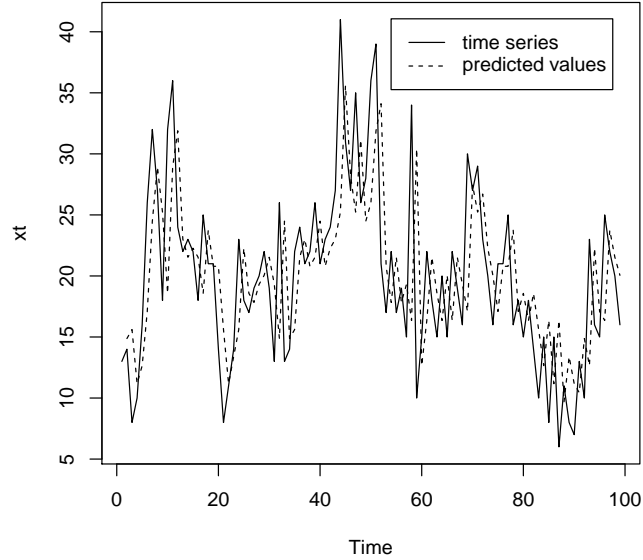} \vspace{-0.5cm}
\caption{\footnotesize{Predicted values and time series plot of earthquakes per year magnitude 7.0 or greater.}} \label{p12}
\end{figure}

\subsection{The number of measles cases by month and notifications rates}
The second example assumes the number of measles cases by month and notifications rates (cases per million) Aug 2013-Dec 2016 in Sweden. Time series plot, autocorrelation and partial autocorrelation functions are shown in Figure \ref{p21}. Sample mean, variance and autocorrelation are respectively, 1.244, 3.489 and 0.35.
 \begin{figure}[!h]
\centering
\includegraphics[height=11.5cm,width=15cm]{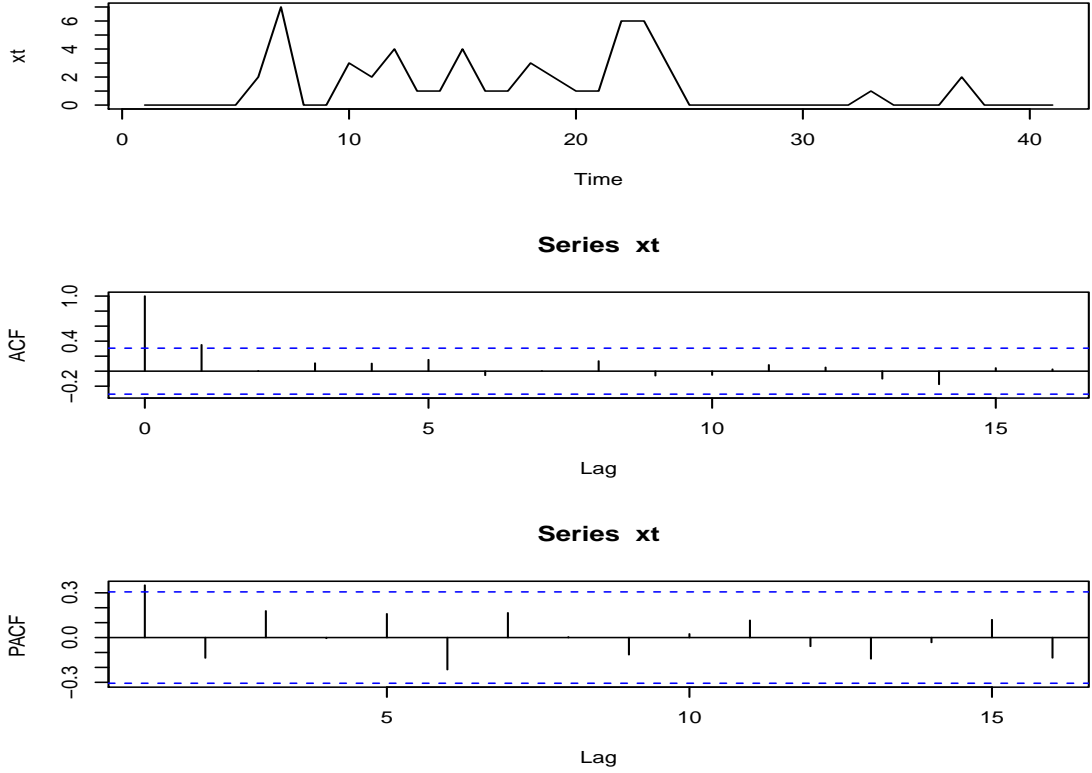} \vspace{-0.5cm}
\caption{\footnotesize{The time series, ACF and PACF plots of the number of measles cases by month and notifications rates (cases per million) Aug 2013-Dec 2016 in Sweden.}} \label{p21}
\end{figure}

The results are shown in Table  \ref{tab2}. According to Table  \ref{tab2}, we see that the AIC and BIC of NBINARPL(1) is smaller than the AIC and BIC of other models and hence, the NBINARPL(1) gives the best fit to this data in comparing with PINARPL(1), BINARPL(1), INARP(1) and INARG(1). So, NBINARPL(1) model with  $W_t \sim PL(2.349)$ innovations that gives by
$$X_t = 0.56 * X_{t-1} +W_t,$$
is more appropriate for this data. The predicted values of the number of measles cases are given by
\begin{eqnarray*}
\hat{X}_1 &=& \frac{\hat{\theta} +2 }{\hat{\theta} (\hat{\theta} +1 )(1 - \hat{\alpha})} =1.256,\\
\hat{X}_i &=& \hat{\alpha} \hat{X}_{i-1} + \frac{\hat{\theta}  +2 }{\hat{\theta} (\hat{\theta} +1 )} =
 0.56 \hat{X}_{i-1} + 0.55 ,\quad\quad i=2,3,...,41.
\end{eqnarray*}
Figure  \ref{p22} shows the predicted values and time series to the sample paths of measles cases.

 \begin{table}[!h]
\begin{center}
{\scriptsize
\caption{\footnotesize{Estimated parameters, AIC and BIC for the number of measles cases by month and notifications rates (cases per million).} }\label{tab2}
\begin{tabular}{ccccccc} \hline
Model &     CLS &     YW & MLE &  AIC &   BIC

\\
\hline
  NBINARPL(1) &    $  \hat{\alpha } =0.355 $ &   $  \hat{\alpha } =0.351$ & $  \hat{\alpha }=0.5631 $   &   122.764& 126.191 \\

 & $  \hat{\theta }  = 1.671     $ &    $  \hat{\theta }=  1.698 $  & $  \hat{\theta }= 2.3490$ & & \\
&&&&& \\
BINARPL(1) & $  \hat{\alpha } =0.355 $ &   $  \hat{\alpha } =0.351$ & $  \hat{\alpha }=0.2506 $   &  125.7234 & 129.1505         \\

 & $  \hat{\theta }  = 1.671     $ &    $  \hat{\theta }= 1.698 $  & $  \hat{\theta }= 1.4874$ & & \\
&&&&& \\
PINARPL(1) & $  \hat{\alpha } =0.355 $ &   $  \hat{\alpha } =0.351$ & $  \hat{\alpha }=0.3307 $   &  124.8738 & 128.3010 \\

 & $  \hat{\theta }  = 1.671     $ &    $  \hat{\theta }= 1.698 $  & $  \hat{\theta }= 1.6368$ & & \\
 &&&&& \\
INARG(1) &$  \hat{\alpha } =0.355 $ &   $  \hat{\alpha } =0.351$ & $  \hat{\alpha }=0.248 $   & 124.623   &  128.050        \\

 & $  \hat{P }  = 0.549     $ &    $  \hat{P }=  0.553 $  & $  \hat{P }= 0.510$ & & \\
&&&&& \\
INARP(1) & $  \hat{\alpha } =0.355 $ &   $  \hat{\alpha } =0.351$ & $  \hat{\alpha }=0.294  $   &     143.801   & 147.228      \\

 & $  \hat{\lambda }  =0.822     $ &    $  \hat{\lambda }=  0.807$  & $  \hat{\lambda }= 0.899$ & & \\
\hline
\end{tabular}}
\end{center}\vspace{-0.5cm}
\end{table}

\begin{figure}[!h]
\centering
\includegraphics[height=9cm,width=9cm]{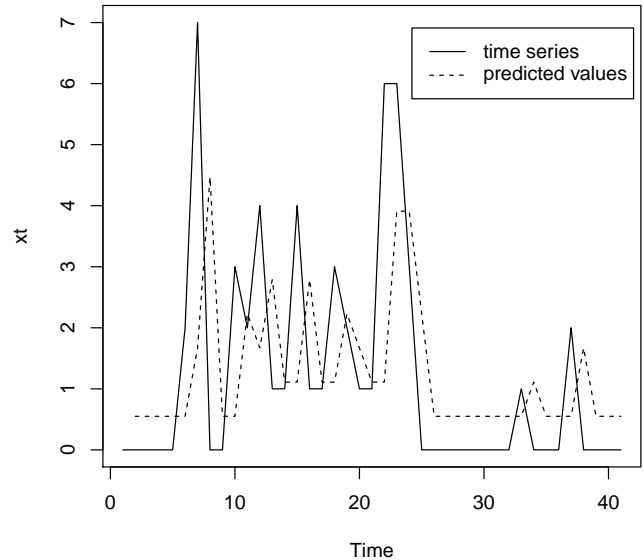} \vspace{-1cm}
\caption{\footnotesize{Predicted values and time series of measles cases.}} \label{p22}
\end{figure}
\newpage

\subsection{The numbers of Sudden death series}
This example assumes the numbers of submissions to animal health laboratories, monthly 2003-2009, from a region in New Zealand. The submissions can be categorized in various ways. Data set is Sudden death series and this data is used by Aghababaei Jazi et al. \cite{Aghababaei2012b}.\\
Time series plot, autocorrelation and partial autocorrelation functions are shown in Figure \ref{p31}. Sample mean, variance and autocorrelation are respectively, 2.0238, 6.529 and 0.59.
 \begin{figure}[!h] \vspace{-0.5cm}
\centering
\includegraphics[height=11cm,width=15cm]{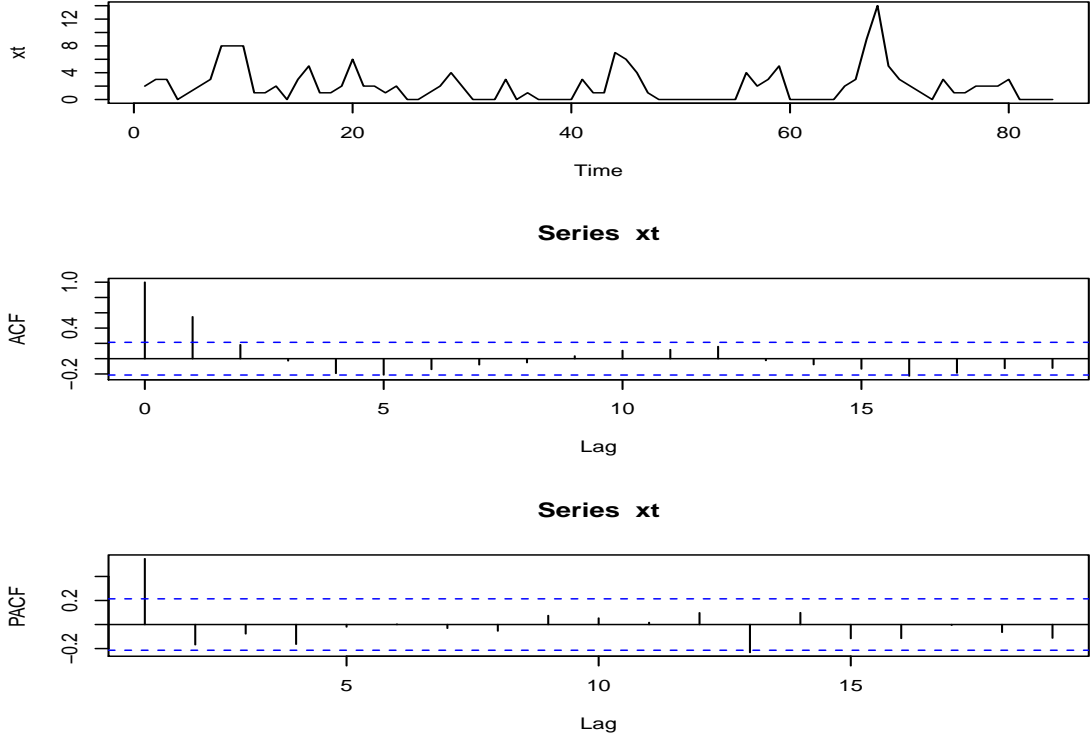} \vspace{-1.5cm}
\caption{\footnotesize{The time series, ACF and PACF plots of Sudden death submissions.}} \label{p31}\vspace{-0.5cm}
\end{figure}

Now, for data modelling, we compare four models.
The results are shown in Table  \ref{tab3}. According to Table  \ref{tab3}, we see that the AIC and BIC of NBINARPL(1) is smaller than the AIC and BIC of other models and hence, the NBINARPL(1) gives the best fit to this data in comparing with PINARPL(1), BINARPL(1), INARP(1) and INARG(1). So, NBINARPL(1) model with $W_t \sim PL(1.6850)$ innovations that gives by
$$X_t = 0.59 * X_{t-1} +W_t ,$$
is more appropriate for this data. The predicted values of the number of Sudden death series series are given by
\begin{eqnarray*}
\hat{X}_1 &=& \frac{\hat{\theta} +2 }{\hat{\theta} (\hat{\theta} +1 )(1 - \hat{\alpha})} =1.986,\\
\hat{X}_i &=& \hat{\alpha} \hat{X}_{i-1} + \frac{\hat{\theta}  +2 }{\hat{\theta} (\hat{\theta} +1 )} =
 0.59\hat{X}_{i-1} + 0.814 ,\quad\quad i=2,3,...,84.
\end{eqnarray*}

 \begin{table}[!h]\vspace{-0.5cm}
\begin{center}
{\scriptsize
\caption{\footnotesize{Estimated parameters, AIC and BIC for the Sudden death submissions.} }\label{tab3}
\begin{tabular}{ccccccc} \hline
Model &     CLS &     YW & MLE &  AIC &   BIC\\
\hline
  NBINARPL(1) &    $  \hat{\alpha } =0.5521 $ &   $  \hat{\alpha } =0.5478$ & $  \hat{\alpha }=0.5888  $   & 297.5909 & 302.4525 \\

 & $  \hat{\theta }  = 1.5221    $ &    $  \hat{\theta }=  1.5255 $  & $  \hat{\theta }= 1.6850$ & & \\
&&&&& \\
BINARPL(1) & $  \hat{\alpha } =0.5521 $ &   $  \hat{\alpha } =0.5478$ & $  \hat{\alpha }=0.3191$&  308.3543 &   313.2159 \\

 &$  \hat{\theta }  = 1.5221    $ &    $  \hat{\theta }=  1.5255 $  & $  \hat{\theta }= 1.0883$ & & \\
&&&&& \\
PINARPL(1) & $  \hat{\alpha } =0.5521 $ &   $  \hat{\alpha } =0.5478$ & $  \hat{\alpha }=0.4732$&  303.1880 &   308.0497 \\

 &$  \hat{\theta }  = 1.5221    $ &    $  \hat{\theta }=  1.5255 $  & $  \hat{\theta }= 1.3599$ & & \\
 &&&&& \\
INARG(1) &$  \hat{\alpha } =0.5521 $ &   $  \hat{\alpha } =0.5478$ & $  \hat{\alpha }=0.317 $   & 306.0826 & 310.9443\\

 & $  \hat{P }  = 0.5215 $ &    $  \hat{P }=  0.5221$  & $  \hat{P }= 0.421$ & & \\
&&&&& \\
INARP(1) & $  \hat{\alpha } =0.5521 $ &   $  \hat{\alpha } =0.5478$ & $  \hat{\alpha }=0.3828$ & 347.4463 &  352.308\\

 & $  \hat{\lambda }  =0.9174  $ &    $  \hat{\lambda }=  0.9151$  & $  \hat{\lambda }= 1.240 $ & & \\
\hline
\end{tabular}}
\end{center}\vspace{-0.5cm}
\end{table}

Figure \ref{p32} shows the predicted values to the sample paths of Sudden death series.

\begin{figure}[!h] \vspace{-1cm}
\centering
\includegraphics[height=8.5cm,width=9cm]{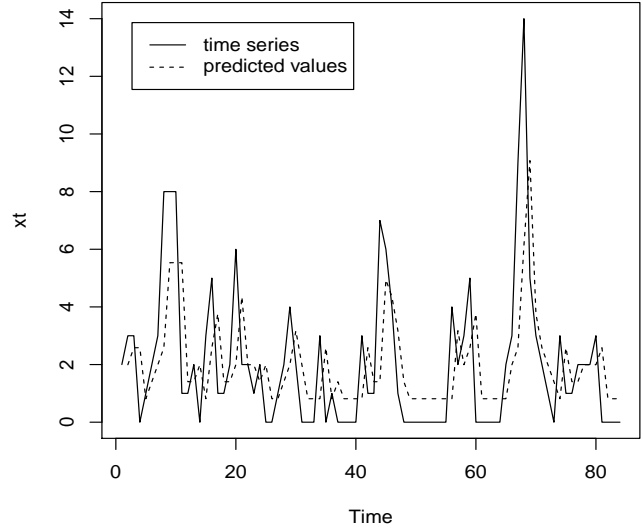} \vspace{-1cm}
\caption{\footnotesize{Predicted values and time series of Sudden death submissions.}} \label{p32}
\end{figure}

\subsection{Weekly counts of the incidence of acute febrile muco-cutaneous lymph node syndrome (MCLS)}
The last example assumes weekly counts of the incidence of acute febrile muco-cutaneous lymph node syndrome (MCLS) in Totori-prefecture, Japan, during 1982.
Time series plot, autocorrelation and partial autocorrelation functions are shown in Figure  \ref{p41}. Sample mean, variance and autocorrelation are respectively, 1.711, 3.111 and 0.5.
 \begin{figure}[!h]
\centering
\includegraphics[height=11.5cm,width=15cm]{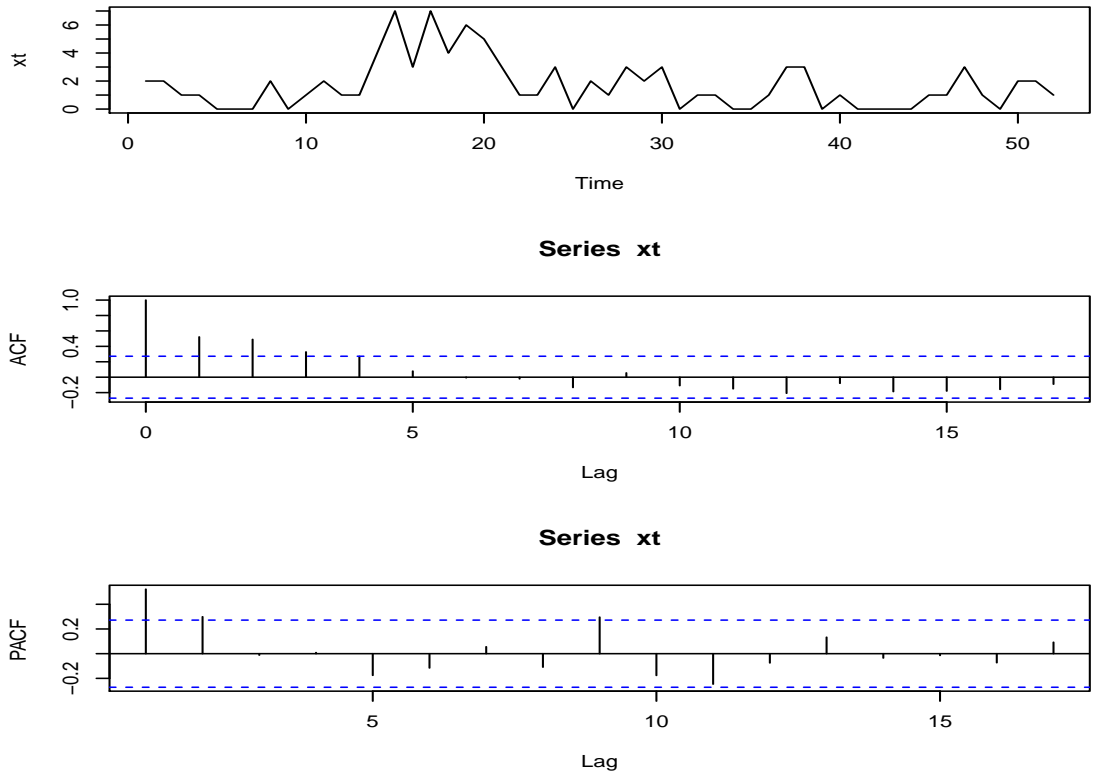} \vspace{-0.5cm}
\caption{\footnotesize{The time series, ACF and PACF plots of the Weekly counts of the incidence of acute febrile muco-cutaneous lymph node syndrome (MCLS) in Totori-prefecture, Japan, during 1982.}} \label{p41}
\end{figure}

The results are shown in Table  \ref{tab4}. According to Table  \ref{tab4}, we see that the AIC and BIC of NBINARPL(1) is smaller than the AIC and BIC of other models and hence, the NBINARPL(1) gives the best fit to this data in comparing with PINARPL(1), BINARPL(1), INARP(1) and INARG(1). So, NBINARPL(1) model with $W_t \sim PL(1.691)$ innovations that gives by
$$X_t = 0.521 * X_{t-1} +W_t ,$$
is more appropriate for this data. The predicted values of the number of polio cases are given by
\begin{eqnarray*}
\hat{X}_1 &=& \frac{\hat{\theta} +2 }{\hat{\theta} (\hat{\theta} +1 )(1 - \hat{\alpha})} =1.693,\\
\hat{X}_i &=& \hat{\alpha} \hat{X}_{i-1} + \frac{\hat{\theta}  +2 }{\hat{\theta} (\hat{\theta} +1 )} =
 0.521\hat{X}_{i-1} + 0.811 ,\quad\quad i=2,3,...,52.
\end{eqnarray*}
Figure \ref{p42} shows the predicted values to the sample paths of weekly counts of the incidence of acute febrile muco-cutaneous lymph node syndrome (MCLS) in Totori-prefecture.

 \begin{table}[!h]
\begin{center}
{\scriptsize
\caption{\footnotesize{Estimated parameters, AIC and BIC for the weekly counts of the incidence of acute febrile muco-cutaneous lymph node syndrome (MCLS) in Totori-prefecture.} }\label{tab4}
\begin{tabular}{ccccccc} \hline
Model &     CLS &     YW & MLE &  AIC &   BIC

\\
\hline
  NBINARPL(1) &    $  \hat{\alpha } =0.524 $ &   $  \hat{\alpha } =0.522$ & $  \hat{\alpha }=0.5209 $   & 170.6369  & 174.5394  \\
 & $  \hat{\theta }  = 1.640 $ &    $  \hat{\theta }=  1.680 $  & $  \hat{\theta }= 1.6908$ & & \\
&&&&& \\
BINARPL(1) & $  \hat{\alpha } =0.524 $ &   $  \hat{\alpha } =0.522$ & $  \hat{\alpha }=0.3832 $   & 172.2558  &176.1583  \\

 & $  \hat{\theta }  = 1.640     $ &    $  \hat{\theta }= 1.680 $  & $  \hat{\theta }= 1.3607$ & & \\
 &&&&& \\

PINARPL(1) & $  \hat{\alpha } =0.524 $ &   $  \hat{\alpha } =0.522$ & $  \hat{\alpha }=0.4804 $   & 171.0987  &175.0012  \\

 & $  \hat{\theta }  = 1.640     $ &    $  \hat{\theta }= 1.680 $  & $  \hat{\theta }= 1.5773$ & & \\
 &&&&& \\
INARG(1) & $  \hat{\alpha } =0.524 $ &   $  \hat{\alpha } =0.522$  & $  \hat{\alpha }=0.3905 $& 172.5549 & 176.4574\\

 & $  \hat{P }  = 0.543     $ &    $  \hat{P }=  0.550 $  & $  \hat{P }= 0.492$ & & \\
&&&&& \\
INARP(1) &  $  \hat{\alpha } =0.524 $ &   $  \hat{\alpha } =0.522$ & $  \hat{\alpha }=0.372  $   & 176.4462  &  180.3487   \\

 & $  \hat{\lambda }  =0.841     $ &    $  \hat{\lambda }=  0.817$  & $  \hat{\lambda }= 1.063$ & & \\
\hline
\end{tabular}}
\end{center}
\end{table}
\begin{figure}[!h]
\centering
\includegraphics[height=9cm,width=9cm]{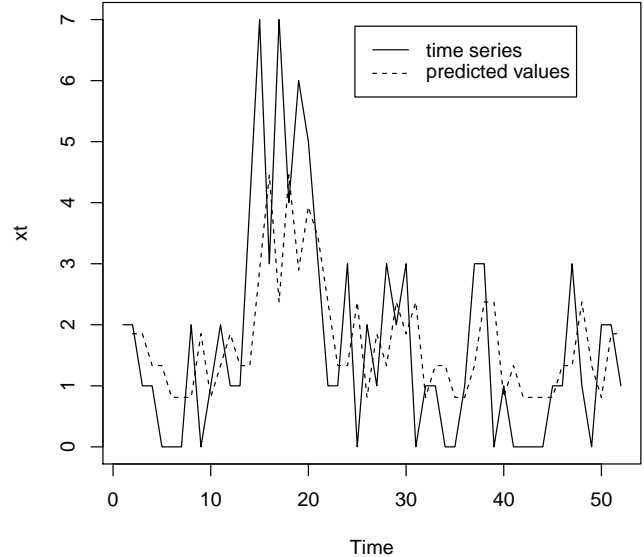} \vspace{-0.5cm}
\caption{\footnotesize{Predicted values and time series of weekly counts of the incidence of acute febrile muco-cutaneous lymph node syndrome (MCLS) in Totori-prefecture.}} \label{p42}
\end{figure}

\newpage
\section{Conclusion}
Integer-valued time series models are very applicable in many fields such as medicine, reliability theory, precipitation, transportation, hotel accommodation and queuing theory. So far, many integer-valued autoregressive process have been introduced by researchers.\\
In this paper we introduce a new stationary first-order integer-valued AR(1) process with Poisson-Lindley innovations based on power series thinning operator. Some mathematical features of these processes are given and estimating the parameters is discussed. Some special cases of this model  (INARPL(1) based on binomial operator,  INARPL(1) based on Poisson operator and INARPL(1) based on negative binomial operator) are studied in some detail.  Finally, some numerical results are presented with a discussion to the obtained results and we fitted PSINARPL(1) model to four real data sets  to show the potentially of the new proposed model.


\end{document}